\documentclass{article}
\usepackage{amssymb}
\usepackage{amsmath,amsfonts,amsthm,amscd,upref,amstext,mathtools,comment}
\usepackage{tikz}
\usepackage{scalefnt}
\usetikzlibrary{trees}
\usetikzlibrary{positioning,quotes}
\usepackage{cancel}
\usepackage{graphicx}
\usepackage{caption}
\usepackage{subcaption}
\usepackage{cancel}
\usepackage[colorlinks=true,urlcolor=blue, citecolor=red,linkcolor=blue,linktocpage,pdfpagelabels, bookmarksnumbered,bookmarksopen]{hyperref}
\usepackage[hyperpageref]{backref}
\usepackage[noabbrev]{cleveref}
\usepackage[shortlabels]{enumitem}
\allowdisplaybreaks

\newtheorem{theorem}{Theorem}
\newtheorem{lemma}[theorem]{Lemma}
\newtheorem{corollary}[theorem]{Corollary}

\newtheorem{observation}{Observation}

\theoremstyle{definition}

\newtheorem{definition}[theorem]{Definition}

\numberwithin{theorem}{section} \numberwithin{equation}{section}

\usepackage{xcolor}  

\usepackage[colorinlistoftodos,prependcaption,textsize=tiny]{todonotes}

\usepackage[colorlinks=true,urlcolor=blue, citecolor=red,linkcolor=blue,linktocpage,pdfpagelabels, bookmarksnumbered,bookmarksopen]{hyperref}
\usepackage[hyperpageref]{backref}
\usepackage{cleveref}

\newenvironment{keyword}{\begin{flushleft}\textbf{KEYWORDS}\\}{\end{flushleft}}

\begin{document}

\title{An $O(\log n)$-Competitive Posted-Price Algorithm for Online Matching on the Line \thanks{Supported in part by NSF grants  CCF-1907673,  CCF-2036077, CCF-2209654 and an IBM Faculty Award.}}

\author{Stephen Arndt \thanks{Computer Science Department, University of Pittsburgh, Pittsburgh PA, 15260. sda19@pitt.edu} \and Josh Ascher\thanks{Computer Science Department, University of Pittsburgh, Pittsburgh PA, 15260. joa71@pitt.edu} \and Kirk Pruhs\thanks{Computer Science Department, University of Pittsburgh, Pittsburgh PA, 15260. kirk@cs.pitt.edu}}

\maketitle

\begin{abstract}
Motivated by demand-responsive parking pricing systems, we consider posted-price algorithms for the online metric matching problem. We give an $O(\log n)$-competitive posted-price randomized algorithm in the case that the metric space is a line. In particular, in this setting we show how to implement the ubiquitous  guess-and-double 
technique using prices.
\end{abstract}


\section{Introduction}

In this paper we are generally interested in addressing a particular  difficulty that arises in 
the design of posted-price algorithms, which is a type of online algorithm that uses 
prices to incentive clients to take actions that increase the social good. Namely, we are interested in the  ``guess and double'' 
 technique that is ubiquitous in the online algorithms literature~\cite{borodin2005online}, but 
 is challenging to implement with prices. In particular we will address this difficulty
 within the context of the problem of online metrical matching on a line metric, 
 with the hope that the algorithmic techniques that we develop will be of use
 in addressing this difficulty in the setting of other online problems. 
 Before giving more details, we need to give some background information. 

As a motivating application for online metric matching, and for posted-price algorithms,
let us consider SFpark, which is San Francisco’s system for managing the availability of on-street parking \cite{SFParkHome,SFParkWikipedia,ACCESSMagazine_2018a}. The goal of SFpark is to reduce the time and fuel wasted by drivers searching for an open parking spot. The system monitors parking usages using sensors embedded in the pavement and distributes this information in real-time to drivers via SFpark.org and phone apps. SFpark periodically adjusts parking meter pricing to manage demand, to lower prices in under-utilized areas, and to raise prices in over-utilized areas. Several other cities in the world have similar demand-responsive parking pricing systems. For example, Calgary has had the ParkPlus system since 2008 \cite{CalgaryPark}.



The problem of centrally assigning drivers to parking spots to minimize time and fuel usage may be
reasonably modeled by the online metric matching problem.  The setting for this 
problem is a collection of servers $S = \{s_1, \dots, s_n\}$ (the parking spots) located at various locations in a metric space. In the case that the metric space is a line, we name the servers so that $s_1 \le s_2 \ldots \le s_n$. Over time a sequence $R = \{r_1, \dots, r_n\}$ of requests (the cars) arrive at various
locations in the metric space.
Upon the arrival of each request (car) $r_i$, the online algorithm must irrevocably be assigned $r_i$ to an available server (parking spot) $s_{\sigma(i)}$, which results in $s_{\sigma(i)}$ being unavailable going forward. Conceptually think of the request (car) $r_i$ moving to server (parking spot) $s_{\sigma(i)}$. 
Thus the cost incurred by such an assignment is the distance
$d(s_{\sigma(i)},r_i)$ between the location of $s_{\sigma(i)}$ and the location where $r_i$ arrived. 
The objective is to minimize the total cost of matching the requests (cars) to
the servers (parking spots).

However, in order to be implementable within the context of SFpark, online algorithms must be posted-price algorithms. In this setting, posted-price means that before each car arrives, the algorithm sets a price on each available parking spot without knowing the next car's arrival location. We assume each car is driven by   a selfish agent who moves to
the available  parking spot that minimizes the sum of the price of that
parking spot and the distance to that parking spot. The objective remains to minimize the aggregate distance traveled by the cars.
It is important to note that conceptually the objective of the parking pricing agency is minimizing
social cost (or equivalently maximizing social good), not maximizing revenue.

Research into posted-price algorithms for online metrical matching was initiated in \cite{cohen2015pricing}, as part of a line of research to study the use of posted-price 
algorithms to minimize social cost in online optimization problems. 
As a posted-price algorithm is a valid online algorithm, one cannot expect to obtain a better competitive ratio for posted-price algorithms than what is achievable
by online algorithms. So this research line has primarily focused on problems where
the optimal competitive ratio achievable by an online algorithm is (perhaps approximately) known,
and seeks to determine whether a similar competitive ratio can be
(again perhaps approximately) achieved by a posted-price algorithm. The higher-level goal is to determine the
increase in social cost that is necessitated by the restriction that an algorithm
has to use posted prices to incentivize selfish agents, instead of being able to mandate agent behavior.

Essentially all results in the posted-price online algorithms literature use one of two 
algorithmic design techniques.
The simpler algorithmic design paradigm is  called {\em mimicry}. A posted-price algorithm $\mathcal{A}$ {\em mimics}
an online algorithm $\mathcal{B}$ if the probability that $\mathcal{B}$ will take a particular action is equal to the 
probability that a self-interested
agent will choose this same action when the prices of actions are set using $\mathcal{A}$. 
However, many online algorithms are not mimickable. So another algorithmic design paradigm,  
called {\em monotonization}, first seeks to identify a sufficient property
for an online algorithm to be  mimickable, and then seeks  to design  an online algorithm with this property. 
In all the examples in the literature, the identified property involves some sort of monotonicity in the behavior of the algorithm. 
In particular, for online metric matching
on a tree metric (which includes a line as a special case), an online algorithm $\mathcal{A}$ is mimickable if and only if it is monotone in
the sense that as the request location moves closer to the location of an available server the
probability that the request is matched to that server cannot decrease~\cite{competitive-pricing}.


There are three online algorithms for online metric matching on a line that interest us here:
\begin{itemize}
    \item The Robust Matching (RM) algorithm is a deterministic primal-dual algorithm
    that is $\Theta(\log n)$-competitive~\cite{optimal-matching-raghvendra}.
     The Robust Matching algorithm is not mimickable~\cite{MaxBender}, and intuitively seems far from being mimickable.
    \item The Harmonic (H) algorithm is a randomized algorithm that is $\Theta(\log \Delta)$-competitive,
    where $\Delta$ is the ratio of the distance between the furthest pair of servers and the distance between the closest pair of servers~\cite{harmonic-alg}.
    The Harmonic algorithm chooses between the first available server to the left of the request and the
    first available server to the right of the request with probability inversely proportional to the distance from
    the request to these servers.
    \cite{cohen2015pricing} showed that the  Harmonic algorithm
is mimickable, thus obtaining an $O(\log \Delta)$-competitive posted-price
algorithm.
    \item The Doubled Harmonic (DH) algorithm is a randomized algorithm that is $O(\log n)$-competitive. Doubled Harmonic combines 
   a variation of Harmonic that uses an estimation $Z$ of the optimal cost (between the requests and the servers), 
   with a standard guess-and-double technique for maintaining a good estimate of the current
   optimal cost to date~\cite{harmonic-alg}. We show in \Cref{sec:DHNotMimickable} that
Doubled Harmonic is not mimickable.
    \end{itemize}

Thus the specific research question that we address is whether we can 
design a monotone variation of Doubled Harmonic that is $O(\log n)$-competitive,
thus leading to an $O(\log n)$-competitive posted-price algorithm.
But, even though it is the title of the paper, 
obtaining a better competitive ratio is only a secondary motivation for this research.
Our primary motivation is to determine whether in this setting we can
implement guess-and-double monotonically, with the hope that this will provide
insights into designing posted-price algorithms in other settings where the
standard online algorithms  use the ubiquitous guess-and-double technique.
To understand why answering this research question isn't completely straightforward, we need to first understand the Doubled Harmonic algorithm.

Firstly, for ease of presentation, we will make some simplifying assumptions, namely:
\begin{itemize}
    \item No pair of servers is closer than 1 unit of distance from 
    each other. We show that this is without loss of generality in \Cref{subsec:one_serv_per_loc}.
    \item All requests arrive at the location of some server. 
     We show that this is without loss of generality in \Cref{subsec:asmp_req_at_serv}.
\end{itemize}
Intuitively Doubled Harmonic modifies Harmonic in following ways~\footnote{Technically our description of Doubled Harmonic differs in some ways from how it is described in \cite{harmonic-alg}, but we believe that
our description is a bit simpler, and the same analysis holds.}.
Firstly, if the distance between consecutive servers is small (less than $Z/n^2$), where $Z$ is
the estimate of optimal maintained by the algorithm, then this distance is artificially inflated (to $Z/n^2$). Secondly, if the actual optimal cost between the requests and servers becomes at least the estimate $Z$, then the estimate $Z$ is increased geometrically until it exceeds the current optimal cost, and the algorithm conceptually reruns itself on all the requests
to date with this new estimate to compute which servers it would ideally like to be available now.
The algorithm then continues forward imagining these servers are available, and then 
correcting to the actually available servers using some optimal matching between the imaginary 
available servers and the actually available servers. Unfortunately the full algorithm, with
corner cases, is a bit more complicated.

\begin{definition}
We define the pseudo-distance $pd\left(s_i, s_{i+1}\right)$ between two adjacent servers $s_i$ and $s_{i+1}$
to be $\infty$ if $s_{i+1} - s_i \ge Z$, 
to be $Z/n^2$ if $s_{i+1} - s_i \le Z/n^2$,
and $s_{i+1} - s_i $ otherwise; here $Z$ will be a parameter in the algorithms.
We then define the pseudo-distance between two arbitrary servers
$s_i$ and $s_j$, where $i < j$ to be
$\sum_{h=i}^{j-1} pd\left(s_h, s_{h+1}\right)$.
\end{definition}

\begin{definition}[Doubled Harmonic Algorithm Description]\label{def:dh}
\end{definition}
Until a request arrives at a location where there is not an available server,
the request is assigned to the available server where it arrives.
When the first request $r_t$ arrives at a location where
there isn't an available server,  the Doubled Harmonic algorithm maintains the following invariants:
\begin{itemize}
    \item 
An estimate $Z= 10^j$, for some integer $j$, such that optimal cost
to date is at least $Z/10$ and is strictly less than $Z$.
\item
A set of 
imaginary servers $S_\iota = \{s_{\iota(1)}, \ldots s_{\iota(k)} \}$ 
that in some sense the algorithm imagines are available (but which may or may
not actually be available).  $S_\iota$ 
is initialized to $S - \{s_{\sigma(1)}, \ldots, s_{\sigma(t-1)}\}$.
\item The set  $S_\rho =\{s_{\rho(1)}, \ldots s_{\rho(k)} \}$ 
of servers that are really available. 
\item 
An arbitrary optimal matching $M$ between $S_\iota $ and $S_\rho$.
\end{itemize}

Then it responds to the arrival of a request $r_t$ in the following way:
\begin{itemize}

\item 
If $r_t$ is triggering, meaning that it causes the optimal cost to date to be at least $Z$, then the   estimate $Z$ is set to $10^j$ where $j$ is
the minimum integer that will reestablish the invariant on $Z$, and
the algorithm then performs what we call an adjustment operation (which we define below) up through request
$r_{t-1}$. 
\item If there is an imaginary server $s_{\iota(i)}$ at the location of $r_t$ then no action is taken (later we will think of this as an imaginary move of length 0).
\item If there is no imaginary server to the left of $r_t$ then it moves
to the first imaginary server to its right. This is called an imaginary move.
\item Else if there is no imaginary server to the right of $r_t$ then it moves
to the first imaginary server to its left. This is called an imaginary move.
    \item 
Else let $s_{\iota(h)}$  and
$s_{\iota(h+1)}$ be the first imaginary servers to the left and right of $r_t$, respectively. 
Then $r_t$ moves to $s_{\iota(h)}$ with probability $$L(s_{\iota(h)}, r_t, s_{\iota(h+1)}) = \frac{ pd(r_t, s_{\iota(h+1)})}{pd(r_t, s_{\iota(h)}) + pd(r_t, s_{\iota(h+1)})}$$
and $r_t$ moves to $s_{\iota(h+1)}$ with probability $$R(s_{\iota(h)}, r_t, s_{\iota(h+1)}) = \frac{ pd(r_t, s_{\iota(h)})}{pd(r_t, s_{\iota(h)}) + pd(r_t, s_{\iota(h+1)})}$$
So the algorithm chooses between the imaginary server to the left and the
imaginary server to the right with probability inversely proportional to the
pseudo-distance. Let us call this movement imaginary movement.
\item 
After the imaginary movement of the request to
a server in $s_{\iota(j)} \in S_\iota$, the request 
continues moving to the server in $s_{\rho(h)}  \in S_\rho $ that $s_{\iota(j)}$ is matched to in $M$,
which we  call  a  corrective move, and $s_{\iota(j)}$ is removed from  $ S_\iota$.
\end{itemize}

\begin{definition}[Adjustment Operation Description]\label{def:adjustment-op}
\end{definition}
This algorithm takes as input a
request $r_t$. The algorithm simulates Doubled Harmonic on all requests
up to $r_t$, sets $S_{\iota}$ to be the servers that would be available
at the end of this simulation, and recomputes an optimal matching $M$.

There are two reasons why modifying Doubled Harmonic to be monotone isn't straightforward (and presumably why this wasn't done in \cite{cohen2015pricing}):
\begin{enumerate}
    \item The first is that the behavior of the algorithm is quite different depending on whether the
    new request is triggering or not, which is challenging to implement with prices because the prices have to be set before the location of the request is known.
    \item The correction moves used by Doubled Harmonic are intuitively not coordinated with the imaginary moves.
\end{enumerate}

Our main contribution is an algorithm that we call Modified Doubled Harmonic (MDH) that
circumvents these issues by modifying Doubled Harmonic in the following way: 
\begin{enumerate}
    \item Triggering requests $r_t$ are just assigned as though they had appeared at a location $x$ near $r_t$ where $r_t$ would not have been triggering had it arrived at location $x$. Intuitively because triggering requests are rare, it's not particularly critical that they
    be handled cheaply.
    \item During the correction step the request moves in same direction as it would in Doubled
    Harmonic, but stops at the first available server.
    Note that this correction step cannot be implemented by any fixed
matching, as Doubled Harmonic does.
\end{enumerate}

One big hurdle in naturally extending poly-log competitiveness results on posted-price algorithms for online metric matching on a  spider metric~\cite{competitive-pricing,spidermatch} 
to tree metrics  is the seeming need to be able to implement guess-and-double in a monotonic way on a tree,
which was the main motivation for considering how to accomplish this on a line~\cite{MaxBender}.
So our takeaway is that this result suggests trying to design the correction step for a tree to be as flexible as possible, so as to make it as easy as possible to monotonically blend with the imaginary movement.

\subsection{Additional Related Work}\label{prior-work} 

Online metric matching was first studied in \cite{onlineweightedmatching,khullermitchell}, and each showed independently that $(2n-1)$-competitive is the optimal competitive ratio for deterministic algorithms in a general metric space. The best known competitive ratio for a randomized algorithm against an oblivious adversary is $O\left(\log^2 n\right)$ \cite{rand-alg-min-matching,random-O(log2k)}, and the best known lower bound is $\Omega(\log n)$. 


In this paper, we focus on matching on the line, which is perhaps the most interesting case. \cite{matching-on-a-line} gave the first deterministic, $o(n)$-competitive algorithm for this problem. \cite{elias-OML} showed that the Generalized Work Function algorithm is $\Omega(\log n)$ and $O(n)$ competitive. \cite{peserico2020matching} showed that no randomized algorithm can achieve a competitive ratio of $o\left(\sqrt{\log n}\right)$ for online matching on the line.

\cite{FeldmanFR17} shows how to set prices to mimic the $O(1)$-competitive  algorithm Slow-Fit from \cite{Aspnes1997,AzarKPPW97} for the problem of minimizing makespan on related machines. 
Monotonization is used in \cite{ImMPS17} to obtain an $O(1)$-competitive posted-price algorithm for minimizing maximum flow time on related machines. 

\section{Modified Doubled Harmonic Description}\label{mdh-description}

We explain the Modified Doubled Harmonic algorithm mainly in terms
of how it differs from Doubled Harmonic. 
Modified Doubled Harmonic makes the same initial assumptions about the
instance, and maintains the same invariants, as does Doubled Harmonic.
Intuitively Modified Doubled Harmonic modifies Doubled Harmonic in the following ways.
Firstly, it handles a triggering request (by pretending it arrived at a nearby point where the
request wouldn't have been triggering if it arrived there) before doing the double step of
guess-and-double. Secondly, during the correction step the request moves in same direction as it would in Doubled
    Harmonic, but stops at the first available server. Unfortunately the details of both
    of these two modifications are a bit complicated.

Note that the optimal matching $M$ between $S_\iota$ and $S_\rho$  partitions the real line into subintervals of three
different types:
\begin{description}
    \item[Left Islands] are maximal subintervals that contain points $x$ where an $s_{\iota(j)} \in S_{\iota}$ to the
    right of $x$ is matched to a  $s_{\rho(h)} \in S_{\rho}$ to the left of $x$ in $M$.
    \item [Right Islands] are maximal subintervals that contain points $x$ where an $s_{\iota(j)} \in S_{\iota}$ to the
    left of $x$ is matched to a  $s_{\rho(h)} \in S_{\rho}$ to the right of $x$ in $M$.
    \item[Stationary Islands] are maximal subintervals that are disjoint from left and right islands.
\end{description}
Note that this partitioning will be the same for all choices of $M$~\cite{optimal-matching-raghvendra}.

\begin{definition}[Modified Doubled Harmonic]\label{def:mdh}\end{definition} 
The algorithm behaves the same way as Doubled Harmonic up until the
first request that arrives at the location of an unavailable server.
The algorithm responds to the arrival of a subsequent request $r_t$ in the following manner:
\begin{enumerate}
\item If $r_t$ appears at the location of a   available server $s_{\rho(j)}$, then it is assigned to $s_{\rho(j)}$. 
\item Else if $r_t$ appears to the left of the leftmost  available server $s_{\rho(1)}$, 
then it is assigned to $s_{\rho(1)}$. 
\item
Else if $r_t$ appears to the right of the rightmost  available server $s_{\rho(k)}$, then it is assigned to $s_{\rho(k)}$. 
\item Else if $r_t$ is not triggering,
\begin{enumerate}
    \item If $r_t$ appears in a left island, it is assigned to the first  available server to its left.
    \item Else if $r_t$ appears in a right island, it is assigned to the first  available server to its right.
    \item Else let $s_{\iota(h)}$  and
$s_{\iota(h+1)}$ be the first imaginary servers to the left and right of $r_t$, respectively. 
Then $r_t$ moves to the first available  server to its left with probability $$L(s_{\iota(h)}, r_t, s_{\iota(h+1)}) = \frac{ pd(r_t, s_{\iota(h+1)})}{pd(r_t, s_{\iota(h)}) + pd(r_t, s_{\iota(h+1)})}$$
and $r_t$ moves to the first available  server to its right with probability $$R(s_{\iota(h)}, r_t, s_{\iota(h+1)}) = \frac{ pd(r_t, s_{\iota(h)})}{pd(r_t, s_{\iota(h)}) + pd(r_t, s_{\iota(h+1)})}$$
So the algorithm chooses between the imaginary server to the left and the
imaginary server to the right with probability inversely proportional to the
pseudo-distance, and then moves to the nearest  available server in that direction.
\end{enumerate}

\item Else (Comment: $r_t$ is triggering)
\begin{enumerate}
\item 
Let $s_{\rho(h)}$  and
$s_{\rho(h+1)}$ be the first available  servers to the left and right of $r_t$, respectively.
\item Let $y_\ell$ be defined in the following way: If one moves from $r_t$ to the left, let $y_\ell$ be the first
point $x$ that one comes to where either $r_t$ would not have been triggering
if it had arrived at $x$, or $x$ is the location of $s_{\rho(h)}$.
\item Let $y_r$ be defined in the following way:
If one moves from $r_t$ to the right, let $y_r$ be the first
point $x$ that one comes to where either $r_t$ would not have been triggering
if it had arrived at $x$, or $x$ is the location of $s_{\rho(h+1)}$.
\item 
Let $m$ be the midpoint between $s_{\rho(h)}$  and
$s_{\rho(h+1)}$.
\item If $R(s_{\rho(h)}, y_r, s_{\rho(h+1)}) < \frac{1}{2}$ then mimic the assignment of a request appearing at $y_r$.
    \item Else if $R(s_{\rho(h)}, y_\ell, s_{\rho(h+1)}) > \frac{1}{2}$
then mimic the assignment of a request appearing at $y_\ell$.
 \item Else if $r_t < m$ then mimic the assignment of a request appearing at $y_\ell$.
    \item Else $r_t \geq m$, and mimic the assignment of a request appearing at $y_r$.
\end{enumerate}
\item If $r_t$ was triggering (this could happen in Cases 1, 2, 3, or 5), 
the algorithm updates the estimate $Z$ and calls the adjustment operation
up through request $r_t$ (note the adjustment operation was defined
when we defined Doubled Harmonic).
\end{enumerate}

To show Modified Doubled Harmonic is well-defined, we make the following observations.

\begin{observation}\label{obs:well_defined}
The following hold for Case 4 of the definition of Modified Doubled Harmonic.
\begin{enumerate}[(a)]
    \item If $r_t$ appears in a left island, then it has an available server to its left. 
    \item If $r_t$ appears in a right island, then it has an available server to its right.
    \item If $r_t$ appears in a stationary island, then there are imaginary servers on each side of $r_t$.
\end{enumerate}

\end{observation}

\begin{proof}
The first two observations follow directly from the definitions of Left Island and Right Island. The third observation follows from the fact that $r_t$ has available servers on each side, and so it must have imaginary servers on each side.
\end{proof}

\section{Monotonicity Analysis}\label{sec:mono_analysis} 
Note that Modified Doubled Harmonic is a  \textit{neighbor} algorithm, that is it always assigns requests to a neighboring server.
In \Cref{lemma:neighbor_monotone} we show that if a neighbor algorithm is monotone on intervals between adjacent available servers $\left(s_{\rho(i)}, s_{\rho(i+1)}\right)$ then it is monotone. In \Cref{lemma:mono_non_trigger} we analyze the probability of a non-triggering request in  $\left(s_{\rho(i)}, s_{\rho(i+1)}\right)$ being assigned to $s_{\rho(i+1)}$.
In \Cref{lemma:constant_trigger} we analyze the probability of a triggering request in  $\left(s_{\rho(i)}, s_{\rho(i+1)}\right)$ being assigned to $s_{\rho(i+1)}$.
Then we conclude in Theorem \ref{thm:mdh-monotone}  that Modified Doubled Harmonic is monotone on each interval $\left(s_{\rho(i)}, s_{\rho(i+1)}\right)$.

Let $r_t \rightarrow s_{\rho(j)}$ denote the event that request $r_t$ is matched to $s_{\rho(j)}$.  We will use the notation $r_t = x$ as shorthand for $r_t$ arrived at location $x$.  We say a point $x$ on the line is a trigger point if a request arriving at location  $x$ would be a triggering request, and otherwise we say $x$ is a non-trigger point.

\begin{lemma}\label{lemma:neighbor_monotone}
A neighbor algorithm $\mathcal{A}$ is monotone if, 
for all intervals of adjacent available servers $\left(s_{\rho(i)}, s_{\rho(i+1)}\right)$, $\Pr\left[r_t \xrightarrow{\mathcal{A}} s_{\rho(i+1)} \ | \ r_t = x\right]$ is non-decreasing across $\left(s_{\rho(i)}, s_{\rho(i+1)}\right)$.
\end{lemma}

\begin{proof}
Suppose for all intervals of adjacent available servers \newline$\left(s_{\rho(i)}, s_{\rho(i+1)}\right)$, $\Pr\left[r_t \rightarrow s_{\rho(i+1)} \ | \ r_t = x\right]$ is non-decreasing across $\left(s_{\rho(i)}, s_{\rho(i+1)}\right)$. Let $u, v, s_{\rho(j+1)} \in \mathbb{R}^1$ be arbitrary such that $v \in [u, s_{\rho(j+1)}]$ and $\mathcal{A}$ has an available server at $s_{\rho(j+1)}$. We want to show the following monotonicity condition holds:

\[ \Pr[r_t \rightarrow s_{\rho(j+1)} \ | \ r = u] \leq \Pr[r_t \rightarrow s_{\rho(j+1)} \ | \ r = v] \]

We proceed by simple casework. If $u = v$, then we have equality; and if $v = s_{\rho(j+1)}$, then $\Pr[r_t \rightarrow s_{\rho(j+1)} \ | \ r_t = v] = 1$. Thus it remains to consider $v \in \left(u, s_{\rho(j+1)}\right)$. If $s_{\rho(j)} \in [u, s_{\rho(j+1)})$, then $\Pr[r_t \rightarrow s_{\rho(j+1)} \ | \ r_t = u] = 0$. Otherwise, if there does not exist an available server to the left of $u$, then $\Pr[r_t \rightarrow s_{\rho(j+1)} \ | \ r_t = v] = 1$. Thus, it remains to consider the case where $u, v \in \left(s_{\rho(j)}, s_{\rho(j+1)}\right)$ for adjacent available servers at $s_{\rho(j)}, s_{\rho(j+1)}$. We know $\Pr\left[r_t \rightarrow s_{\rho(j+1)} \ | \ r_t = x\right]$ is non-decreasing across this interval, and so we must have $\Pr[r_t \rightarrow s_{\rho(j+1)} \ | \ r_t = u] \leq \Pr[r \rightarrow s_{\rho(j+1)} \ | \ r_t = v]$. Thus in all cases, the monotonicity condition holds. 
If instead we pick $u, v, s_{\rho(j+1)} \in \mathbb{R}^1$ arbitrary with $v \in [s_{\rho(j+1)}, u]$, the same reasoning holds. Thus the described condition implies $\mathcal{A}$ is monotone, and so it is equivalent to monotonicity for neighbor algorithms.

Let $r_t \xrightarrow{\text{MDH}} s_{\rho(j)}$ denote the event that request $r_t$ is matched to available server $s_{\rho(j)}$ using Modified Doubled Harmonic. We now fix an arbitrary interval of adjacent available servers $\left(s_{\rho(i)}, s_{\rho(i+1)}\right)$.
\end{proof}

\begin{lemma}\label{lemma:mono_non_trigger}
$\Pr\left[r \xrightarrow{\emph{MDH}} s_{\rho(i+1)} \ | \ r_t = x\right]$ is non-decreasing across the non-trigger points in $\left(s_{\rho(i)}, s_{\rho(i+1)}\right)$.
\end{lemma}

\begin{proof}
    Note that the interval $\left(s_{\rho(i)}, s_{\rho(i+1)}\right)$ can be expressed as the union of a left island, a stationary island, and a right island (any two of which could possibly be empty). Since \cite{optimal-matching-raghvendra} guarantees they must appear in this order, the fact that MDH assigns a request $r_t$ in a stationary island to $s_{\rho(i+1)}$ with probability inversely proportional to its pseudodistance from $s_{\rho(i+1)}$ yields the result. 
\end{proof}

\begin{lemma}\label{lemma:constant_trigger}
For all subintervals $\left(x_L, x_R\right) \subseteq \left(s_{\rho(i)}, m\right) \cup \left(m, s_{\rho(i+1)}\right)$ containing only trigger points, where $m$ is the midpoint of $\left(s_{\rho(i)}, s_{\rho(i+1)}\right)$, we have $\Pr\left[r_t \xrightarrow{\emph{MDH}} s_{\rho(i+1)} \ | \ r_t = x\right]$ is constant across $\left(x_L, x_R\right)$.
\end{lemma}

\begin{proof}
Let $(x_L,x_R)\subseteq \left(s_{\rho(i)}, m\right) \cup \left(m, s_{\rho(i+1)}\right)$ containing only trigger points be arbitrary. Note that the only information used to make the assignments of triggering requests are the adjacent non-trigger points (or endpoints of the interval) and the arrival location of the triggering requests relative to the midpoint. Since $(x_L,x_R)$ contains no non-trigger points and is entirely contained on one side of $m$, all of this information is identical. Thus, all requests in $(x_L,x_R)$ have the same probability of being assigned to $s_{\rho(i+1)}$. 
\end{proof}

\begin{theorem}\label{thm:mdh-monotone} Modified Doubled Harmonic is monotone.
\end{theorem}

\begin{proof}
The non-trigger points in $\left(s_{\rho(i)}, s_{\rho(i+1)}\right)$, along with $m$, partition the interval into subintervals for which $\Pr\left[r_t \xrightarrow{\text{MDH}} s_{\rho(i+1)} \ | \ r_t = x\right]$ is constant via \Cref{lemma:constant_trigger}. Further, \Cref{lemma:mono_non_trigger} shows that $\Pr\left[r_t \xrightarrow{\text{MDH}} s_{\rho(i+1)} \ | \ r_t = x\right]$ is non-decreasing across non-trigger points, and Case 5 of \Cref{def:mdh} ensures that the probability of assigning a triggering request to $s_{\rho(i+1)}$ is sandwiched between the probability of assigning its neighboring non-trigger points to $s_{\rho(i+1)}$. So, \Cref{lemma:neighbor_monotone} implies that MDH is monotone. 
\end{proof}

\section{Cost Analysis}\label{subsec:cost_analysis}

In this section we prove Theorem \ref{thm:mdh-competitive}, which states 
that Modified Doubled Harmonic is $O(\log n)$-competitive. 

\begin{theorem} \label{thm:mdh-competitive} 
MDH is $O(\log n)$-competitive for online matching on the line.
\end{theorem}

We first break the execution of Modified Doubled Harmonic into phases, where each
phase terminates with a triggering request. 
We show in Lemma \ref{lemma:ec_invariant} that the aggregate cost of the nontriggering requests during a phase is at most $O(\log n)$ times 
the current estimate of the optimal cost plus the imaginary cost that Doubled Harmonic would have incurred during that phase. We accomplish this by showing that for each nontriggering request, the cost of the optimal matching between the imaginary and available servers decreases by at least the amount that
the cost for Modified Doubled Harmonic exceeds the imaginary cost that Doubled Harmonic would
have incurred on that request. 
In \Cref{lemma:ei_crucial} we bound the cost to Modified Doubled Harmonic for
a triggering request
by twice the greedy cost (which can be seen to be $O(\log n)$ times OPT via the traingle inequality) and the
cost to Modified Doubled Harmonic if the request had arrived at a nearby non-trigger point. Once we have established Lemma
\ref{lemma:ec_invariant}  and Lemma \ref{lemma:ei_crucial}, the bounding of
Modified Doubled Harmonic's cost proceeds as in \cite{harmonic-alg}.

\subsection{Cost Analysis Definitions}
We first need some definitions. Let $S_\iota(t)$ be the set of imaginary servers before the arrival of $r_t$, and let $S_\rho(t)$ be the set of available servers before the arrival of $r_t$. Let $D\left(S_\iota(t), S_\rho(t)\right)$ be the optimal cost of matching $S_\iota(t)$ and $S_\rho(t)$. Let $s_{\sigma(t)}$ be the available server that Modified Doubled Harmonic used
for request $r_t$. For a nontriggering request $r_t$, if $r_t$ appeared in a left island or a right island, let $s_{\gamma(t)}$ be the imaginary server that would be selected if one selected a neighboring imaginary server to either the left or right of $r_t$ with probability inversely proportional to the pseudo-distance. If instead $r_t$ appeared in a stationary island, then if one moves from $r_t$ in the direction of $s_{\sigma(t)}$, let $s_{\gamma(t)}$ be the first imaginary server one hits. Define a phase as the sequence of requests which appear while MDH has the same estimate $Z$ on the optimal cost. Phases begin with a sequence of nontriggering requests, and terminate with a single triggering request, after which the estimate $Z$ inflates. 

Let OPT$(t)$ be the optimal cost of matching the first $t$ requests to the servers, and suppose that OPT$(n) \in \left[10^{\ell}, 10^{\ell+1}\right)$. For ease of presentation, suppose that before the estimate $Z$ is instantiated during execution of MDH, it holds a default value of $1$. Then the estimate $Z$ runs through $Z = 10^{k_i}$ for $0 = k_0 < k_1 < k_2 \dots < k_m = \ell+1$. Let $Z_i = 10^{k_i}$ for each $0 \leq i \leq m$. We now introduce some definitions which allow us to partition the requests according to $Z_i$. Let

\begin{itemize}
    \item $\tau_i$ be the \textbf{maximum index} $t$ such that OPT$(t) < Z_i$ for each $0 \leq i \leq m$.
    \item $\rho_i$ be the $i$'th \textbf{triggering request}, which upon appearance causes OPT$(t)$ to increase from $ < Z_{i-1}$ to $ \geq Z_{i-1}$. Equivalently $\rho_i = r_{\tau_{i-1} + 1}$.
    \item $B_i$ be the \textbf{sequence} of requests $r_t$ arriving after $\rho_i$ and before $\rho_{i+1}$.
\end{itemize}

Let $B_0$ and $B_m$ be the sequence of requests appearing before $\rho_1$ and after $\rho_m$, respectively. This allows us to decompose the full request sequence as $B_0, \rho_1, B_1, \rho_2, \dots, B_{m-1}, \rho_m, B_m$. The phase of the algorithm associated with $Z_i$ is given by the pair $\left(B_i, \rho_{i+1}\right)$. We now introduce some definitions which allow us to partition MDH's assignments and DH's underlying imaginary moves according to $Z_i$. Let

\begin{itemize}
    \item $W_i = \bigcup_{r_t \in B_i} \{\left(r_t, s_{\sigma(t)}\right)\}$ be the set of \textbf{assigned edges} for the requests in $B_i$.
    \item $X_i = \bigcup_{r_t \in B_i} \{\left(r_t, s_{\gamma(t)}\right)\}$ conceptually be the set of chosen \textbf{imaginary moves} for the requests in $B_i$.
\end{itemize}

Conceptually, $X_i$ is a set of possible imaginary moves of Doubled Harmonic. These imaginary moves are relevant for us because we bound the cost of Modified Doubled Harmonic's assignments against the cost of these imaginary moves. We are also interested in how MDH / DH simulates request assignments during an adjustment operation. For this reason, define $s_{\mu(i, t)}$ to be the imaginary server chosen for the request $r_t$ during the adjustment operation triggered by $\rho_i$. Of course, $s_{\mu(i, t)}$ is only defined for $t \leq \tau_{i-1}$, because the adjustment operation which occurs after the estimate inflates to $Z = Z_i$ only simulates request assignments up to the triggering request $\rho_i = r_{\tau_{i-1}+1}$. Now, let

\begin{itemize}
    \item $Y_i = \bigcup_{t=1}^{\tau_{i-1}} \{\left(r_t, s_{\mu(i, t)}\right)\}$ be the set of \textbf{simulated assignments} of the requests for the adjustment operation triggered by $\rho_i$.
    \item $e_i = \{\left(r_{t'}, s_{\sigma(t')}\right)\}$ be the \textbf{assigned edge} of $\rho_i$. Here $t' = \tau_{i-1} + 1$.
    \item $f_i = \{\left(r_{t'}, s_{\gamma(t')}\right)\}$ conceptually be the chosen \textbf{imaginary move} for $\rho_i$. Here $t' = \tau_{i-1} + 1$, and $s_{\gamma(t')}$ is a neighboring imaginary server to $\rho_i$ chosen with probability inversely proportional to the pseudodistance \textit{after} the adjustment operation triggered by $\rho_i$ is performed.
    \item $E_i= \{e_1, e_2, \dots, e_i\}$ be the set of \textbf{assigned edges} for the triggering requests up through $\rho_i$.
\end{itemize}

This allows us to decompose the full assigned edge set $W = \left(\bigcup_{i=0}^m W_i\right) \cup E_m$ in the order $W = W_0, e_1, W_1, e_2, \dots, W_{m-1}, e_m, W_m$. We can further decompose the chosen imaginary moves in the order $X = X_0, f_1, X_1, f_2, \dots, X_{m-1}, f_m, X_m$. Lastly, for an edge set $U$, let $|U|$ be the sum of the lengths of the edges in $U$.

\subsection{Bounding Non-Trigger Costs}

Before proving \Cref{lemma:ec_invariant}, we establish a useful fact about how the cost of the optimal matching between the available and imaginary servers will change during the nontriggering requests of a phase (i.e. during the requests of a $B_i$).

\begin{lemma}\label{lemma:DPQ}
Let $P, Q$ be two finite sets of points in $\mathbb{R}^1$ with the same number of elements. Suppose $P = \{p_1, p_2, \dots, p_m\}$ and $Q = \{q_1, q_2, \dots, q_m\}$, where the points have been written in increasing order of location. Let $D(P, Q)$ be the optimal cost of matching $P$ and $Q$. Further, let $P' = P \setminus \{p_g\}$ and $Q' = Q \setminus \{q_h\}$ for arbitrary $g, h \in [1, m]$. Then

\[   D\left(P', Q'\right) - D(P, Q) \leq \left\{
\begin{array}{ll}
      \left(p_h - p_g\right) - |p_h - q_h| & \ \ g \leq h \\
      \left(q_g - q_h\right) - |q_g - p_g| & \ \ g > h \\
\end{array} 
\right. \]
\end{lemma}

\begin{proof}
We know via \cite{optimal-matching-raghvendra} that

\[ D(P, Q) = \sum_{k=1}^m |p_k - q_k| \]

Suppose $g \leq h$. Then we have

\[ D\left(P', Q'\right) = \sum_{k=1}^{g-1} |p_k - q_k| + \sum_{k=g}^{h-1} |p_{k+1} - q_k| + \sum_{k=h+1}^m |p_k - q_k| \]

Thus

\begin{align*}
D\left(P', Q'\right) - D(P, Q) &= \left(\sum_{k=g}^{h-1} |p_{k+1} - q_k| - |p_k - q_k|\right) - |p_h - q_h| \\
&\leq \left(\sum_{k=g}^{h-1} |p_{k+1} - p_k|\right) - |p_h - q_h| \\
&= \left(\sum_{k=g}^{h-1} \left(p_{k+1} - p_k\right)\right) - |p_h - q_h| \\
&= \left(p_h - p_g\right) - |p_h - q_h|
\end{align*}

For $g > h$, the proof follows identically, only with $P, Q$ and $g, h$ switched.
\end{proof}

\begin{lemma}\label{lemma:ec_invariant} 
Consider an arbitrary phase, and renumber the nontriggering requests $B_i$ in that phase to $r_1, r_2, \dots, r_k$. With probability one the expression
\[ D(S_\rho(t), S_\iota(t)) + \sum_{j=1}^{t-1} \left(d(r_j, s_{\sigma(j)}) - d(r_j, s_{\gamma(j)})\right) \] is a non-increasing  function of $t$.
\end{lemma}

\begin{proof}
Define $g(t)$ to be the above expression for the chosen phase, and let $t \in [1, k]$ be arbitrary. Then we have
\[ g(t+1) - g(t) = D(S_\rho(t+1), S_\iota(t+1)) - D(S_\rho(t), S_\iota(t)) + d(r_t, s_{\sigma(t)}) - d(r_t, s_{\gamma(t)}) \]

where $S_\rho(t+1) = S_
\rho(t) \setminus \{s_{\sigma(t)}\}$ and $S_\iota(t+1) = S_\iota(t) \setminus \{s_{\gamma(t)}\}$. Write $S_\rho(t) = \{s_{\rho(1)}, s_{\rho(2)}, \dots, s_{\rho(\ell)} \}$ and $S_\iota(t) = \{s_{\iota(1)}, s_{\iota(2)}, \dots, s_{\iota(\ell)} \}$ where the servers in each set have been ordered left-to-right. Suppose $s_{\sigma(t)} = s_{\rho(a)}$ and $s_{\gamma(t)} = s_{\iota(b)}$.

Now, suppose $a < b$. Then because $s_{\rho(a)} < s_{\rho(b)}$ and MDH is a neighbor algorithm, we must have $r_t < s_{\rho(b)}$. Further, because $s_{\iota(a)} < s_{\iota(b)}$ and $s_{\iota(b)}$ is a neighboring imaginary server to $r_t$, we must have $r_t > s_{\iota(a)}$. The final observation is the trickiest to notice: $r_t \leq s_{\rho(a)}$, meaning that $a < b$ implies MDH \textit{cannot} assign $r_t$ leftwards. We can show this through simple casework on the description of Modified Doubled Harmonic. The only cases where leftward assignment is possible are Case 3, Case 4a, and Case 4c. However, in all of these cases, we must have $a \geq b$. Indeed, in Case 3, $a = \ell \geq b$. In Case 4a, $r_t$ is in a left island, and so \cite{optimal-matching-raghvendra} shows $r_t$ must have more available servers than imaginary servers on its left, forcing $a \geq b$. In Case 4c, $r_t$ is in a stationary island, and so \cite{optimal-matching-raghvendra} shows there must be an equal number of available and imaginary servers to the left of (and including the location of) $r_t$. By definition of $s_{\gamma(t)}$ when $r_t$ is in a stationary island, we must have $a = b$. Thus given $a < b$, leftward assignment of $r_t$ is not possible, and so $r_t \leq s_{\rho(a)}$. Finally, we can deduce $s_{\iota(a)} < r_t \leq s_{\rho(a)} < s_{\rho(b)}$. Simple computation yields

\begin{align*}
g(t+1) - g(t) &= D(S_{\rho}(t+1), S_{\iota}(t+1)) - D(S_{\rho}(t), S_{\iota}(t)) + d(r_t, s_{\rho(a)}) - d(r_t, s_{\iota(b)}) \\
&\leq \left(s_{\rho(b)} - s_{\rho(a)}\right) - d(s_{\rho(b)}, s_{\iota(b)}) + \left(s_{\rho(a)} - r_t\right) - d(r_t, s_{\iota(b)}) \\
&= \left(s_{\rho(b)} - r_t\right) - \left(d(s_{\rho(b)}, s_{\iota(b)}) + d(r_t, s_{\iota(b)})\right) \\
&= d(s_{\rho(b)}, r_t) - \left(d(s_{\rho(b)}, s_{\iota(b)}) + d(r_t, s_{\iota(b)})\right) \\
&\leq 0
\end{align*}

The first inequality follows from \Cref{lemma:DPQ}. If $a = b$, direct computation gives the same result. If $a > b$, applying the same reasoning as before gives $s_{\rho(b)} < s_{\rho(a)} \leq r_t < s_{\iota(a)}$, and the same result follows. Thus in all cases, $g(t+1) - g(t) \leq 0$ giving $g(t+1) \leq g(t)$. Thus $g(t)$ is a non-increasing function of $t$, completing the proof.
\end{proof}

Let $i \in [0, m]$ be arbitrary. We now pursue the goal of bounding $|W_i|$, the total cost of the non-trigger assignments while Modified Doubled Harmonic has estimate $Z = Z_i$. We start by recalling an important result from \cite{harmonic-alg}, which gives a cost bound on the imaginary moves and the simulated assignments from the adjustment operation.

\begin{lemma}\label{lemma:EV_ideal}
\cite{harmonic-alg} $\mathbb{E}\left[|X_i| + |f_i| + |Y_i|\right] \leq C \cdot Z_i$ for $C = O(\log n)$.
\end{lemma}

Moving forward, we will use $C$ to refer to the specific $O(\log n)$ function which is used in \Cref{lemma:EV_ideal}. Because $|X_i|$ is properly bounded by $O(Z_i \log n)$, our goal now becomes bounding $|W_i| - |X_i|$, the amount Modified Doubled Harmonic exceeds the imaginary cost that Doubled Harmonic would have incurred on the requests in $B_i$.

Let $\hat{t_i} = \tau_{i-1} + 2$ be the time of the first request in $B_i$. To bound $|W_i| - |X_i|$, we will bound $D\left(S_\rho\left(\hat{t_i}\right), S_\iota\left(\hat{t_i}\right)\right)$, which will be sufficient for our purposes upon application of \Cref{lemma:ec_invariant}. We do so by constructing a matching $M_i : S_\iota\left(\hat{t_i}\right) \rightarrow S_\rho\left(\hat{t_i}\right)$ whose cost is appropriately bounded.

\begin{lemma}\label{lemma:costMi} \cite{harmonic-alg}
There exists a bijection $M_i : S_\iota\left(\hat{t_i}\right) \rightarrow S_\rho\left(\hat{t_i}\right)$ such that

\[ \text{\emph{cost}}\left(M_i\right) \leq |f_i| + |Y_i| + |E_i| + \sum_{j=0}^{i-1} |W_j| \]
\end{lemma}

\begin{proof}\cite{harmonic-alg}
Cover the line with $\{f_i\} \cup Y_i$ and $\left(\bigcup_{j=0}^{i-1} W_j\right) \cup E_i$. For all imaginary and available servers at the same location, match them together. Otherwise, for each remaining imaginary server in $S_\iota\left(\hat{t_i}\right)$, follow the edges of this covering until an available server in $S_\rho\left(\hat{t_i}\right)$ is reached, and match them together. Via the triangle inequality, the induced matching $M_i : S_\iota\left(\hat{t_i}\right) \rightarrow S_\rho\left(\hat{t_i}\right)$ has

\[\text{cost}\left(M_i\right) \leq |f_i| + |Y_i| + |E_i| + \sum_{j=0}^{i-1} |W_j| \]
\end{proof}

\begin{lemma}\label{lemma:bound_Wi}
$\mathbb{E}\left[|W_i|\right] \leq C \cdot Z_i + \mathbb{E}\left[|E_i| + \sum_{j=0}^{i-1} |W_j|\right]$.
\end{lemma}

\begin{proof}
$\hat{t_i}$ is the time of the first request in $B_i$, and $\tau_i + 1$ is the time of the $(i+1)$'st triggering request $\rho_{i+1}$. Thus $\hat{t_i} \leq \tau_i + 1$ and so \Cref{lemma:ec_invariant} implies $g\left(\tau_i + 1\right) \leq g\left(\hat{t_i}\right)$. Thus

\[D\left(S_\rho\left(\tau_i + 1\right), S_\iota\left(\tau_i + 1\right)\right) + \left(|W_i| - |X_i|\right) \leq D\left(S_\rho\left(\hat{t_i}\right), S_\iota\left(\hat{t_i}\right)\right)\]

This gives

\begin{align*}
|W_i| &\leq |X_i| + D\left(S_\rho\left(\hat{t_i}\right), S_\iota\left(\hat{t_i}\right)\right) - D\left(S_\rho\left(\tau_i + 1\right), S_\iota\left(\tau_i + 1\right)\right) \\
&\leq |X_i| + D\left(S_\rho\left(\hat{t_i}\right), S_\iota\left(\hat{t_i}\right)\right) \\
&\leq |X_i| + \text{cost}(M_i) \\
&\leq |X_i| + |f_i| + |Y_i| + |E_i| + \sum_{j=0}^{i-1} |W_j|
\end{align*}

The third inequality follows from the fact that $D\left(S_\rho\left(\hat{t_i}\right), S_\iota\left(\hat{t_i}\right)\right)$ is the optimal cost of matching $S_\rho\left(\hat{t_i}\right)$ and $S_\iota\left(\hat{t_i}\right)$. The last inequality follows from \Cref{lemma:costMi}. Applying \Cref{lemma:EV_ideal} gives the desired result.
\end{proof}

\subsection{Bounding Trigger Costs}

We now prove a sequence of lemmas with the eventual goal of proving \Cref{lemma:bound_ei}. We begin by introducing some basic functions to compute assignment costs.
In \Cref{lemma:ei_crucial}, we bound the cost to Modified Doubled Harmonic for a triggering request
by twice the greedy cost (which is clearly $O(\log n)$ times OPT) and the cost to Modified Doubled Harmonic if the request had arrived at a nearby non-trigger point. We bound the greedy cost of $\rho_i$ in \Cref{lemma:greedy_edge}, and the cost bound on the non-trigger points is a simple corollary from \Cref{lemma:bound_Wi}. Combining these results, we prove \Cref{lemma:bound_ei}.

First, we introduce some basic functions for computing assignment costs. The function $L_h(x)$ is the linear transformation of $\left(s_{\rho(h)}, s_{\rho(h+1)}\right)$ onto $(0, 1)$ (which maps $s_{\rho(h)}$ to $0$ and $s_{\rho(h+1)}$ to $1$). The function $N(\alpha, \gamma) = \alpha(1-\gamma) + (1-\alpha)\gamma$ is a ``normalized'' assignment cost, where we assume the adjacent available servers exist at $0$ and $1$. The following lemma makes these ideas rigorous.

\begin{lemma}\label{lemma:NL_funcs}
Suppose request $r_t$ appears in between adjacent available servers $s_{\rho(h)}$ and $s_{\rho(h+1)}$. Further, suppose $r_t$ assigns to $s_{\rho(h)}, s_{\rho(h+1)}$ with probabilities $1-p, p$. Then the expected cost of $r_t$'s assignment is $\left(s_{\rho(h+1)} - s_{\rho(h)}\right) N(L_h(r_t), p)$.
\end{lemma}

\begin{proof}
The proof follows directly from simple computation.
\end{proof}

The utility of decomposing $r_t$'s assignment cost in this way comes from the fact that we may now concern ourselves with studying $N$, the normalized assignment cost, which simplifies much of the computation. Next, we establish some useful facts about the function $N$. Each fact will be used in bounding the cost in each subcase of Case 5 of \Cref{def:mdh}.

\begin{lemma}\label{lemma:N_facts}
The following facts hold for all $\alpha, \beta, \gamma \in [0, 1]$.

\begin{enumerate}[(a)]
    \item If $\alpha \leq \beta$ and $\gamma \leq \frac{1}{2}$, then $N\left(\alpha, \gamma\right) \leq N(\beta, \gamma)$.
    \item If $\alpha \geq \beta$ and $\gamma \geq \frac{1}{2}$, then $N\left(\alpha, \gamma\right) \leq N(\beta, \gamma)$.
    \item If $\beta \leq \alpha \leq \frac{1}{2}$ and $\gamma \leq \frac{1}{2}$, then $N\left(\alpha, \gamma\right) \leq 2\max\left(\alpha, N(\beta, \gamma)\right)$.
    \item If $\beta \geq \alpha \geq \frac{1}{2}$ and $\gamma \geq \frac{1}{2}$, then $N\left(\alpha, \gamma\right) \leq 2\max\left(1 - \alpha, N(\beta, \gamma)\right)$.
\end{enumerate}

\end{lemma}

\begin{proof}
\hfill
\begin{enumerate}[(a)]
    \item This follows directly via simple computation.
    \item This follows directly via simple computation.
    \item \[ \frac{N(\alpha, \gamma)}{\alpha} = \frac{\gamma + \alpha - 2\gamma\alpha}{\alpha} \leq \frac{\gamma + \alpha}{\alpha} = 1 + \frac{\gamma}{\alpha} \]
    \[ \frac{N(\alpha, \gamma)}{N(\beta, \gamma)} = \frac{\gamma + \alpha - 2\gamma\alpha}{\gamma + \beta - 2\gamma\beta} \leq \frac{\gamma + \alpha}{\gamma + \beta(1-2\gamma)} \leq \frac{\gamma + \alpha}{\gamma} = 1 + \frac{\alpha}{\gamma}\]
    Because $\min\left(\frac{\gamma}{\alpha}, \frac{\alpha}{\gamma}\right) \leq 1$, we know $N(\alpha, \gamma) \leq 2\alpha$ or $N(\alpha, \gamma) \leq 2N(\beta, \gamma)$. Thus $N\left(\alpha, \gamma\right) \leq 2\max\left(\alpha, N(\beta, \gamma)\right)$.
    \item Via direct computation, $N(\alpha, \gamma) = N(1-\alpha, 1-\gamma)$ and $N(\beta, \gamma) = N(1-\beta, 1-\gamma)$. Thus upon application of (c), we have
    \begin{align*}
        N(1-\alpha, 1-\gamma) &\leq 2\max(1-\alpha, N(1-\beta, 1-\gamma)) \\
        N(\alpha, \gamma) &\leq 2\max(1-\alpha, N(\beta, \gamma))
    \end{align*}
\end{enumerate}
\end{proof}

\begin{lemma}\label{lemma:ei_crucial}
Consider a triggering request $r_t$. 
Let $s_j$ be the available server closest to $r_t$.
Let $y_{\ell}$ and $y_r$ be defined as in the Modified Doubled Harmonic algorithm,
and let $s_\ell$ and $s_r$ be the available servers that Modified Doubled Harmonic
would have assigned a request arriving at $y_\ell$ and $y_r$, respectively. 
Then 
$$\mathbb{E}\left[ d\left(r_t, s_{\sigma(t)}\right) \right] \le 2\max\left(   \mathbb{E}\left[ d\left(y_\ell, s_\ell\right)  \right], \mathbb{E}\left[ d\left(y_r, s_r\right)  \right], d\left(r_t, s_j\right)
\right)$$
\end{lemma}

\begin{proof} We proceed by showing the claim holds in each potential trigger case of \Cref{def:mdh}. In Cases 1, 2, and 3, the claim trivially holds, because $\rho_i$ is assigned greedily to $s_j$. It remains to consider Case 5. Suppose $\rho_i$ appeared in between adjacent available servers $s_{\rho(h)}$ and $s_{\rho(h+1)}$, and let $m$ be the midpoint of $\left(s_{\rho(h)}, s_{\rho(h+1)}\right)$. Suppose that under the linear transformation $L_h : \left(s_{\rho(h)}, s_{\rho(h+1)}\right) \rightarrow (0, 1)$, $\rho_i$ maps to $\alpha$, $y_\ell$ maps to $\beta_\ell$, and $y_r$ maps to $\beta_r$. Further, $m$ trivially maps to $\frac{1}{2}$. In comparing the costs of assignments in $\left(s_{\rho(h)}, s_{\rho(h+1)}\right)$, it suffices to compare the costs of the normalized assignments in $(0, 1)$, given the normalization factor of $s_{\rho(h+1)} - s_{\rho(h)}$ is always the same.

Let $p_\ell = R(s_{\rho(h)}, y_\ell, s_{\rho(h+1)})$ and $p_r = R(s_{\rho(h)}, y_r, s_{\rho(h+1)})$. \Cref{lemma:N_facts} cleanly handles each subcase of Case 5.

\begin{enumerate}[(a)]
    \item If $p_r < \frac{1}{2}$, then $\rho_i$ mimics the assignment of a request arriving at $y_r$. We know $\alpha \leq \beta_r$, and so $N\left(\alpha, p_r\right) \leq N\left(\beta_r, p_r\right)$.
    \item Else if $p_\ell > \frac{1}{2}$, then $\rho_i$ mimics the assignment of a request appearing at $y_\ell$. We know $\alpha \geq \beta_\ell$, and so $N\left(\alpha, p_\ell\right) \leq N\left(\beta_\ell, p_\ell\right)$.
    \item Else if $\rho_i < m$, then $\rho_i$ mimics the assignment of a request appearing at $y_\ell$. We know $\beta_\ell \leq \alpha \leq \frac{1}{2}$ and $p_\ell \leq \frac{1}{2}$, and so $N\left(\alpha, p_\ell\right) \leq 2\max\left(\alpha, N(\beta_\ell, p_\ell)\right)$. Note that $\alpha$ is simply the normalized greedy assignment of $\rho_i$.
    \item Else $\rho_i \geq m$, and $\rho_i$ mimics the assignment of a request appearing at $y_r$. We know $\beta_r \geq \alpha \geq \frac{1}{2}$ and $p_r \geq \frac{1}{2}$, and so $N\left(\alpha, p_r\right) \leq 2\max\left(1-\alpha, N(\beta_r, p_r)\right)$. Note that $1-\alpha$ is simply the normalized greedy assignment of $\rho_i$.
\end{enumerate}

Given $\rho_i$ assigns rightwards to $s_{\rho(h+1)}$ with probability $p$, in all cases, we have

\[ N(\alpha, p) \leq 2\max\left(N\left(\beta_\ell, p_\ell\right), N\left(\beta_r, p_r\right), \min(\alpha, 1-\alpha)\right) \]

Multiplying both sides by the normalization factor of $s_{\rho(h+1)} - s_{\rho(h)}$ gives the desired result.

\end{proof}

It remains to bound the cost of all individual non-trigger assignments and the greedy assignment. First, we obtain a bound on the greedy assignment of $\rho_i$. 

\begin{lemma}\label{lemma:greedy_edge}
For a triggering request $\rho_i$, let $s$ be the available server nearest to $\rho_i$. Then $\mathbb{E}\left[d\left(\rho_i, s\right)\right] \leq C \cdot Z_i +  \mathbb{E}\left[|E_{i-1}| + \sum_{j=0}^{i-1} |W_j|\right]$.
\end{lemma}

\begin{proof}
Run the adjustment operation on all requests up to $\rho_i = r_{\tau_{i-1}+1}$ to generate simulated assigned servers $s_{\mu(i, t)}'$ and a set of simulated assignments $Y_i' = \bigcup_{t=1}^{\tau_{i-1}} \left\{\left(r_t, s_{\mu(i, t)'}\right)\right\}$ solely for the purposes of our argumentation. This produces a set of imaginary servers $S_\iota'$. Cover the line with the edges in $Y_i'$ and the assigned edges $\left(\bigcup_{j=0}^{i-1} W_j\right) \cup E_{i-1}$. This covering partitions the line into disjoint intervals for which each interval has the same number of requests, servers in $Y_i'$, and previously assigned servers in $\{s_{\sigma(1)}, s_{\sigma(2)}, \dots, s_{\sigma\left(\tau_{i-1}\right)}\}$. By extension, each partition must have the same number of imaginary servers in $S_\iota'$ and available servers in $S_\rho$.

Now pick an imaginary server $s_{\gamma(t')}'$ for the triggering request $\rho_i$ in the same way we picked $s_{\gamma(t')}$, where here $t' = \tau_{i-1}+1$. This gives a generated imaginary move $f_i' = \left\{\left(r_{t'}, s_{\gamma(t')}'\right)\right\}$, and add $f_i'$ to this covering. From $\rho_i$, follow $f_i'$, reaching a (previously) imaginary server $s_{\iota(g)}' \in S_\iota'$. Some available server must exist within the partition containing $s_{\iota(g)}'$, and so the triangle inequality ensures that some available server exists at most distance $|f_i'| + |Y_i'| + |E_{i-1}| + \sum_{j=0}^{i-1} |W_j|$ from $\rho_i$. Given $s$ is the available server nearest to $\rho_i$, we must have

\begin{align*}
\mathbb{E}\left[d\left(\rho_i, s\right)\right] &\leq \mathbb{E}\left[|f_i'| + |Y_i'| + |E_{i-1}| + \sum_{j=0}^{i-1} |W_j|\right] \\
&\leq C \cdot Z_i + \mathbb{E}\left[|E_{i-1}| + \sum_{j=0}^{i-1} |W_j|\right]
\end{align*}

where in the final step we apply \Cref{lemma:EV_ideal}.
\end{proof}

Next, we obtain a bound on assignments of requests appearing at non-trigger points, which is a direct corollary from \Cref{lemma:bound_Wi}.

\begin{corollary}\label{cor:Wi}
For a non-trigger point $y$, let $s$ be the available server that Modified Doubled Harmonic would have assigned a request arriving at $y$, given the estimate is currently $Z = Z_{i-1}$. Then $\mathbb{E}\left[d(y, s)\right] \leq C \cdot Z_{i-1} + \mathbb{E}\left[|E_{i-1}| + \sum_{j=0}^{i-2} |W_j|\right]$.
\end{corollary}

With all of the pieces in place, we establish a cost bound on $\mathbb{E}\left[|e_i|\right]$.

\begin{lemma}\label{lemma:bound_ei}
$\mathbb{E}\left[|e_i|\right] \leq 2\left(C \cdot Z_i +  \mathbb{E}\left[|E_{i-1}| + \sum_{j=0}^{i-1} |W_j|\right]\right)$. 
\end{lemma}

\begin{proof}
The proof follows directly from application of \Cref{lemma:ei_crucial}, \Cref{cor:Wi}, and \Cref{lemma:greedy_edge}. Note that we apply \Cref{cor:Wi} when the estimate is $Z = Z_{i-1}$ because $\rho_i$ causes the estimate to inflate from $Z = Z_{i-1}$ to $Z = Z_i$.
\end{proof}

\subsection{Proving \Cref{thm:mdh-competitive}}

We now make the recursive bounds on $\mathbb{E}\left[|W_i|\right]$, $\mathbb{E}\left[|e_i|\right]$ established in \Cref{lemma:bound_Wi}, \Cref{lemma:bound_ei} explicit through induction. The key idea is that although $\mathbb{E}\left[|W_i|\right]$ and $\mathbb{E}\left[|e_i|\right]$ are bounded in terms of all previous assignments and imaginary moves, the geometrically increasing nature of $Z_i$ ensures their costs are simply on the order of $C \cdot Z_i$.

\begin{lemma}\label{lemma:induct_Wiei}
$\mathbb{E}\left[|W_i|\right] \leq 8C \cdot Z_i$ and $\mathbb{E}\left[|e_i|\right] \leq 5C \cdot Z_i$ for all $i \in [0, m]$.
\end{lemma}

\begin{proof}
First, note that for all $i \in [0, m-1]$,

\[ \sum_{j=0}^i Z_j = \sum_{j=0}^i 10^{k_j} \leq \sum_{h=0}^{k_i} 10^h = \frac{1}{9} \cdot \left(10^{k_i+1} - 1\right) \leq \frac{1}{9} \cdot 10^{k_{i+1}} = \frac{1}{9} \cdot Z_{i+1} \]

We now proceed by induction on $i$. The base case of $|W_0| = 0$ is trivial, and simply define $|e_0| = 0$. Let $i \in [0, m-1]$ be arbitrary, and assume the claim holds for all $j \in [0, i]$. Then

\begin{align*}
\mathbb{E}\left[|e_{i+1}|\right] &\leq 2\left(C \cdot Z_{i+1} +  \mathbb{E}\left[|E_i| + \sum_{j=0}^i |W_j|\right]\right) \\
&= 2C \cdot Z_{i+1} + 2\cdot\sum_{j=0}^i \mathbb{E}\left[|e_j|\right] + 2\cdot\sum_{j=0}^i \mathbb{E}\left[|W_j|\right] \\
&\leq 2C \cdot Z_{i+1} + 10C\cdot\sum_{j=0}^i Z_j + 16C\cdot\sum_{j=0}^i Z_j \\
&= 2C \cdot Z_{i+1} + 26C\cdot\sum_{j=0}^i Z_j \\
&\leq 2C \cdot Z_{i+1} + \frac{26C}{9} \cdot Z_{i+1} \\
&\leq 5C \cdot Z_{i+1}
\end{align*}

and 

\begin{align*}
\mathbb{E}\left[|W_{i+1}|\right] &\leq C \cdot Z_{i+1} + \mathbb{E}\left[|E_{i+1}| + \sum_{j=0}^i |W_j|\right] \\
&= C \cdot Z_{i+1} + \sum_{j=0}^{i+1} \mathbb{E}\left[|e_j|\right] + \sum_{j=0}^i \mathbb{E}\left[|W_j|\right] \\
&\leq C \cdot Z_{i+1} + 5C \cdot \sum_{j=0}^{i+1} Z_j + 8C \cdot \sum_{j=0}^i Z_j \\
&= 6C \cdot Z_{i+1} + 13C \cdot \sum_{j=0}^i Z_j \\
&\leq 6C \cdot Z_{i+1} + \frac{13C}{9} \cdot Z_{i+1} \\
&\leq 8C \cdot Z_{i+1}
\end{align*}

completing the induction.
\end{proof}

Finally, we now prove \Cref{thm:mdh-competitive}. The $O(\log n)$-competitiveness of Modified Doubled Harmonic is a direct consequence of \Cref{lemma:induct_Wiei} and the fact the geometric sums are asymtotically equal
to their largest summand.

\begin{proof}[Proof of \Cref{thm:mdh-competitive}]
Aggregating the edges $W = \left(\bigcup_{i=0}^m W_i\right) \cup E_m$, we have

\[\mathbb{E}\left[|W|\right] = \mathbb{E}\left[\left(\sum_{i=0}^m |W_i|\right) + |E_m|\right] \\
= \sum_{i=0}^m \mathbb{E}\left[|W_i|\right] + \sum_{i=1}^m \mathbb{E}\left[|e_i|\right]\]
Applying \Cref{lemma:induct_Wiei},
\[
\sum_{i=0}^m \mathbb{E}\left[|W_i|\right] + \sum_{i=1}^m \mathbb{E}\left[|e_i|\right]
\leq 8C \cdot \sum_{i=0}^m Z_i + 5C \cdot \sum_{i=0}^m Z_i = 13C \cdot \sum_{i=0}^m Z_i
\]
Simplifying yields
\[
13C \cdot \sum_{i=0}^m Z_i \leq \frac{13C}{9} \cdot 10^{k_m + 1} 
\leq 1.5C \cdot 10^{\ell + 2} 
= 150C \cdot 10^{\ell}
\leq 150C \cdot \text{OPT}(n)  
\]
Recalling $C = O(\log n)$ completes the proof.

\end{proof}

\bibliographystyle{abbrv}
\bibliography{bib.bib}

\newpage

\appendix

\section{Remedying Some Assumptions}\label{app:asmp}

\subsection{Minimum Distance 1 Between Servers }\label{subsec:one_serv_per_loc}
First, note that we may always assume the minimum distance between servers at \textit{different locations} is $1$, which can be easily remedied by a suitable scaling. Thus to resolve the assumption that the minimum distance between adjacent servers is $1$, the important piece to resolve is that no two servers exist at the same location. 

Suppose we have a monotone
neighbor algorithm $\mathcal{A}$ which is $\alpha$-competitive under the assumptions that all servers exist at different locations, and requests appear at server locations. We will construct a monotone neighbor algorithm $\mathcal{B}$ which is $2\alpha$-competitive and removes the first assumption. 

Again we may assume the instance given to $\mathcal{B}$ has minimum distance 1 between adjacent servers at different locations, which can be easily remedied by a suitable scaling. We do this primarily for ease of analysis. Let $\epsilon = \frac{1}{5n}$. On the instance given to $\mathcal{B}$, construct an instance for $\mathcal{A}$ by first placing one server per server location; and then perturbing the extra servers at the same location by at most $\epsilon$ (so that all servers are now at distinct locations). $\mathcal{B}$ then services request $r$ in the following way.

\begin{itemize}
    \item If $r$ appears at an available server $s$ in the instance of $\mathcal{B}$, place a simulated request $\tilde{r}$ at an available server $\tilde{s}$ in the same ``$\epsilon$-window'' in the instance of $\mathcal{A}$. Then $r \xrightarrow{\mathcal{B}} s$ and $\tilde{r} \xrightarrow{\mathcal{A}} \tilde{s}$.
    \item Otherwise, let $t$ be the location of $r$'s appearance. Place a simulated request $\tilde{r}$ at $t$ in the instance of $\mathcal{A}$. Given $\tilde{r} \xrightarrow{\mathcal{A}} s$ for an available server $s$, then $r \xrightarrow{\mathcal{B}} s$.
\end{itemize}

It is easy to see that $\mathcal{B}$ is a monotone neighbor algorithm given $\mathcal{A}$ is a monotone neighbor algorithm. It remains to show $\mathcal{B}$ is $2\alpha$-competitive. Note that each assignment in $\text{ON}_\mathcal{B}$ and $\text{OPT}_\mathcal{B}$ can differ from the corresponding assignment in $\text{ON}_\mathcal{A}$ and $\text{OPT}_\mathcal{A}$ by at most $\epsilon$. Thus 

\[\text{ON}_\mathcal{B} \leq \text{ON}_\mathcal{A} + n\epsilon = \text{ON}_\mathcal{A} + \frac{1}{5}\] 

and 

\[\text{OPT}_\mathcal{B} \geq \text{OPT}_\mathcal{A} - n\epsilon = \text{OPT}_\mathcal{A} - \frac{1}{5}\]

If $\text{OPT}_\mathcal{B} = 0$, then $\text{ON}_\mathcal{B} = 0$, because all requests appeared at available servers. Otherwise, $\text{OPT}_\mathcal{B} > 0$, and so some request is forced to match to a server at a different location. Because the minimum distance between adjacent servers (at different locations) is $1$, we must have $\text{ON}_\mathcal{B} \geq \text{OPT}_\mathcal{B} \geq 1$. The same property holds for the instance of $\mathcal{A}$ (where some request is forced to assign outside of its ``$\epsilon$-window''), and so $\text{ON}_\mathcal{A} \geq \text{OPT}_\mathcal{A} \geq 1 - 2\epsilon \geq \frac{3}{5}$. Thus

\[\frac{\mathbb{E}\left[\text{ON}_\mathcal{B}\right]}{\text{OPT}_\mathcal{B}} \leq \frac{\mathbb{E}\left[\text{ON}_\mathcal{A}\right] + \frac{1}{5}}{\text{OPT}_\mathcal{A} - \frac{1}{5}} \leq \left(\frac{1 + \frac{1}{3}}{1 - \frac{1}{3}}\right)\left(\frac{\mathbb{E}\left[\text{ON}_\mathcal{A}\right]}{\text{OPT}_\mathcal{A}}\right) \leq 2\alpha \]

and so $\mathcal{B}$ is $2\alpha$-competitive, as desired.

\subsection{Requests Appear at Server Locations}\label{subsec:asmp_req_at_serv}
Suppose we have a monotone neighbor algorithm $\mathcal{B}$ which is $\beta$-competitive under the assumption that requests appear at server locations. We will construct a monotone neighbor algorithm $\mathcal{C}$ which is $(2\beta + 1)$-competitive and makes no such assumption. Specifically, $\mathcal{C}$ services request $r$ in the following way.

\begin{itemize}
    \item Let $t$ be the server closest to $r$, regardless of whether $t$ is available or not.
    \item Place a simulated request $\tilde{r}$ at the location of $t$ in the running instance of $\mathcal{B}$.
    \item Given $r \xrightarrow{\mathcal{B}} s$ for an available server $s$, then $r \xrightarrow{\mathcal{C}} s$.
\end{itemize}

Let $S = \{s_1, s_2, \dots, s_n\}$ be the set of servers in the instance and $R = \{r_1, r_2, \dots, r_n\}$ be the set of requests. Without loss of generality, assume the servers of $S$ and the requests of $R$ have been written, according to their locations, in increasing order of coordinate value. Let $t_i$ be the server nearest to $r_i$, regardless of whether it is available or not upon appearance of $r_i$. Then the set $T = \{t_1, t_2, \dots, t_n\}$ is written as ``ordered'' as well.

First, we show $\mathcal{C}$ is $(2\beta + 1)$-competitive. Suppose $\mathcal{B}$ assigns $\tilde{r}_i$ to $s_{\sigma(i)}$ for each $i$. Then OPT$_{\mathcal{B}} = \sum_{i=1}^n d(t_i, s_i)$, ON$_{\mathcal{B}} = \sum_{i=1}^n d(s_{\sigma(i)}, t_i)$, OPT$_{\mathcal{C}} = \sum_{i=1}^n d(r_i, s_i)$, and ON$_{\mathcal{C}} = \sum_{i=1}^n d(s_{\sigma(i)}, r_i)$, where the structure of OPT$_{\mathcal{B}}$ and OPT$_{\mathcal{C}}$ is given by \cite{optimal-matching-raghvendra}. Note OPT$_{\mathcal{C}} \geq \sum_{i=1}^n d(r_i, t_i)$ since $t_i$ is the nearest server to $r_i$ for each request $r_i$. Then we have

\begin{align*}
\text{OPT}_{\mathcal{C}} &= \frac{1}{2} \left(\text{OPT}_{\mathcal{C}} + \text{OPT}_{\mathcal{C}}\right) \\
&\geq \frac{1}{2}\left(\sum_{i=1}^n d(r_i, s_i) + \sum_{i=1}^n d(r_i, t_i)\right) \\
&\geq \frac{1}{2}\left(\sum_{k=1}^n d(t_i, s_i)\right) \\
&= \frac{1}{2}\text{OPT}_{\mathcal{B}}
\end{align*}

and

\[ \text{ON}_{\mathcal{C}} = \sum_{i=1}^n d(s_{\sigma(i)}, r_i) \leq \sum_{i=1}^n d(s_{\sigma(i)}, t_i) + \sum_{i=1}^n d(r_i, t_i) \leq \text{ON}_{\mathcal{B}} + \text{OPT}_{\mathcal{C}} \]

Thus

\[ \frac{\mathbb{E}\left[\text{ON}_{\mathcal{C}}\right]}{\text{OPT}_C} \leq \frac{\mathbb{E}\left[\text{ON}_{\mathcal{B}}\right] + \text{OPT}_C}{\text{OPT}_{\mathcal{C}}} \leq \frac{\mathbb{E}\left[\text{ON}_{\mathcal{B}}\right]}{\frac{1}{2}\text{OPT}_{\mathcal{B}}} + 1 = 2\left(\frac{\mathbb{E}\left[\text{ON}_\mathcal{B}\right]}{\text{OPT}_\mathcal{B}}\right) + 1 \leq 2\beta + 1 \]

as desired. Further, it is easy to see that $\mathcal{C}$ is a neighbor algorithm given $\mathcal{B}$ is a neighbor algorithm. 
Lastly, we must show $\mathcal{C}$ is monotone. Indeed, the sets of points closest to $s_i$ for each server $s_i$ partition the real line into disjoint intervals (where all servers at the same location are understood to share the same interval). Any requests appearing within the same interval are treated identically in $\mathcal{C}$. This discretization ensures that because $\mathcal{B}$ is monotone and thus satisfies the condition in \Cref{lemma:neighbor_monotone}, $\mathcal{C}$ satisfies the same condition, and so it is also monotone.

\section{Proof that Doubled Harmonic is Not Monotone}
\label{sec:DHNotMimickable} 

Consider the following instance. 

{\scalefont{.75}
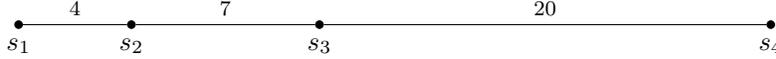
\begin{figure}[h]
\begin{center}
\begin{tikzpicture}[every edge quotes/.style = {auto, font=\footnotesize, sloped}]
\node[draw, shape = circle, fill = black, minimum size = 0.1cm, inner sep=0pt, label={below:$s_1$}] (s_1) at  (1,0)  {};
\node[draw, shape = circle, fill = black, minimum size = 0.1cm, inner sep=0pt, label={below:$s_2$}] (s_2) at  (2.5,0)  {};
\node[draw, shape = circle, fill = black, minimum size = 0.1cm, inner sep=0pt, label={below:$s_3$}] (s_3) at  (5,0)  {};
\node[draw, shape = circle, fill = black, minimum size = 0.1cm, inner sep=0pt, label={below:$s_4$}] (s_4) at  (11,0)  {};
\draw (s_1) edge["4"]  (s_2);
\draw (s_2) edge["7"]  (s_3);
\draw (s_3) edge["20"]  (s_4);
\end{tikzpicture}
\label{fig:DH-mono-violation}
\end{center}
\end{figure}
}

Suppose that $r_1$ arrives at $s_2$. Then, $r_1 \xrightarrow{\mathrm{DH}} s_2$. Next, suppose $r_2$ arrives at $s_2$. Then the optimal matching of $r_1$ and $r_2$ has cost 4, the estimate $Z$ is set to 10, the set of imaginary servers is set to $S_\iota = \{s_1, s_3, s_4\}$, and the set of available servers is set to $S_\rho = \{s_1, s_3, s_4\}$. Clearly the optimal matching $M$ between $S_\iota$ and $S_\rho$ just assigns each server to itself. Suppose DH then performs the imaginary move $r_2 \rightarrow s_1$ and the subsequent corrective move $s_1 \rightarrow s_1$. This leaves $S_\iota = \{s_3, s_4\}$ and $S_\rho = \{s_3, s_4\}$. Now, we show that the assignment of $r_3$ is not monotone.

Suppose that $r_3$ appears at $s_1$. Then, the optimal matching of the requests has cost 7. DH performs the imaginary move $r_3 \rightarrow s_3$ and the subsequent corrective move $s_3 \rightarrow s_3$, and so DH assigns $r_3$ to $s_3$ with probability 1.

Suppose that $r_3$ instead appears at $s_2$. Now, the optimal matching of the requests has cost 11. The estimate $Z$ is then set to 100, and the adjustment operation is performed. With probability $\frac{4}{11}$, DH simulates assigning $r_1$ to $s_2$ and $r_2$ to $s_3$. The imaginary move of $r_3$ is then to $s_1$ with probability $\frac{27}{31}$, and the subsequent corrective move assigns $r_3$ to $s_3$. The imaginary move of $r_3$ to $s_4$ has probability $\frac{4}{31}$, and the subsequent corrective move assigns $r_3$ to $s_4$. Thus with nonzero probability, DH assigns $r_3$ to $s_4$ (and thus NOT to $s_3$) in this case.

Thus the probability that DH assigns $r_3$ to $s_3$ is \textit{higher} for arrival at $s_1$ (probability 1) than for arrival at $s_2$ (probability $< 1$). Further note that this violation of monotonicity is induced by the fact that an adjustment operation will not occur if $r_3$ arrives at $s_1$, but it will occur if $r_3$ arrives at $s_2$. Thus DH is not monotone.

\end{document}